\documentclass[%
 reprint,
 amsmath,amssymb,
 aps,
]{revtex4-2}

\usepackage{graphicx}
\usepackage{dcolumn}
\usepackage{bm}

\usepackage{amsthm,amsmath,amssymb}
\usepackage{mathrsfs}
\usepackage{bm}
\usepackage{lineno}
\newtheorem{definition}{Definition}
\newtheorem{theorem}{Theorem}
\newtheorem{lemma}{Lemma}
\newtheorem{inference}{Inference}

\begin{document}

\preprint{APS/123-QED}

\title{Generalized Potential and Mathematical Principles of Nonlinear Analysis}

\author{Peng Yue}
\email[Emial: ]{pengyu.yue@outlook.com}
\affiliation{University of Electronic Science and technology of China, Chengdu, China}


\begin{abstract}
In the past hundred years, chaos has always been a mystery to human beings, including the butterfly effect discovered in 1963 and the dissipative structure theory which won the chemistry Nobel Prize in 1977. So far, there is no quantitative mathematical-physical method to solve and analyze these problems. In this paper, the idea of using field theory to study nonlinear systems is put forward, and the concept of generalized potential is established mathematically. The physical essence of generalized potential promoting the development of nonlinear field is extended and the spatiotemporal evolution law of generalized potential is clarified.  Then the spatiotemporal evolution law of conservative system and pure dissipative system is clarified. Acceleration field, conservative vector field and dissipation vector field are established to evaluate the degree of conservation and dissipation of physical field. Finally, the development route of new field research and the precondition of promoting engineering application in the future are discussed.
\end{abstract}
\maketitle



For a long time, the researches on nonlinear problems have attracted much attention. The reason for this enthusiasm for research is that, on the one hand, due to the continuous in-depth study of turbulence, buckling and other mechanical theories, the problems found have obvious nonlinear characteristics \cite{Kuchnen2019}, which greatly restricts related theoretical breakthroughs and engineering innovation. On the other hand, the well-known butterfly effect, dissipation mechanism and other factors also show the influence of nonlinear factors in the researches of medical statistics \cite{Shen2020}, economics \cite{Manda2020}, sociology \cite{Barnett2020} and other disciplines, which also makes the research on nonlinear problems attract the attention of researchers in other subjects. Therefore, it is necessary to develop innovative mathematical analysis methods to describe nonlinear phenomena.

At present, in addition to relying on phenomenological statistics to obtain the basic laws of nonlinear dynamics \cite{Kuchnen2019}\cite{Shen2020}\cite{Manda2020}\cite{Barnett2020}, international researches on nonlinear problems are mainly based on two major ideas. One is to rely on the continuous development of mathematical theory, through geometric methods or the solution of nonlinear partial differential equations \cite{Buckmaster2019} to obtain the basic laws of system nonlinearity. The other is to add the necessary physical conditions to linearize the nonlinear problem, and develop and improve the approximate characteristics of nonlinearity through the study of linear problems \cite{Hamzi2019}. However, one obstacle is that the development of mathematical theory depends on sufficient inspiration and continuous updating of philosophical methodology. Another obstacle is that the physical simplification method of linearization will lead to the loss of key information of nonlinear problems. For example, the current common initial value sensitive problems cannot be reflected by linearized models. Based on the above two problems, the nonlinear theory of classical mechanics is already in a bottleneck period, and it is urgent to develop new analysis methods to deal with the current problems.

Usually, any nonlinear system can be expressed by the following general expressions:
\begin{equation}
\dot{x}_{i}=f_{i}\left(x_{j}\right), f_{i}\in C^{2}\left(\mathbb{R}^{n}\right)
\end{equation}

For example, Lorentz system can be expressed as:
\begin{equation}
\begin{array}{c}
\dot{x}_{1}=\sigma\left(x_{2}-x_{1}\right)\\
\dot{x}_{2}=rx_{1}-x_{2}-x_{1}x_{3}\\
\dot{x}_{3}=x_{1}x_{2}-bx_{3}
\end{array}
\end{equation}

At present, the common method for studying above the system is Lagrange method. The above-mentioned system is taken as the motion equation of the particle system, and the stability of the system can be judged by studying the singularity of the system \cite{Kuznetsov2020}. This method is widely used in related projects. However, this method has great limitations in revealing the nonlinear dissipation characteristics of the system at the system level, for example, it is very difficult to judge the occurrence of butterfly effect and reveal the characteristics of dissipative structure in chaos systems.
In the field of mechanics, the common analysis methods include Lagrange method and Euler method. In this paper, Euler method is used to describe the nonlinear system. For example, for Lorentz system (2), “phase flow” can be expressed as FIG. 1 to FIG. 3.
\begin{figure}[ht]
\includegraphics[scale=0.14]{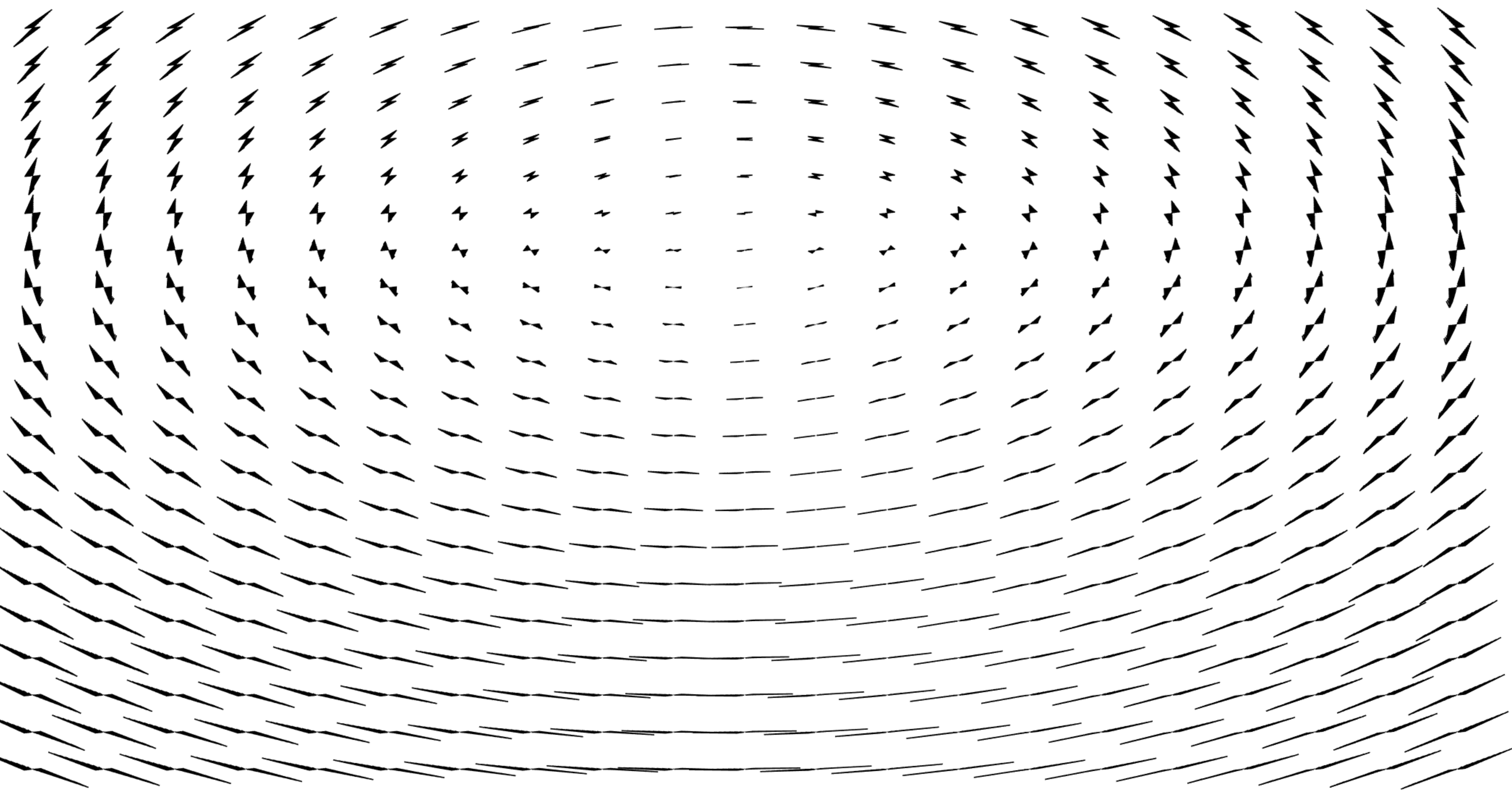}
\caption{\label{}Phase flow diagram on XOY plane for Lorentz system}
\end{figure}
\begin{figure}[ht]
\includegraphics[scale=0.10]{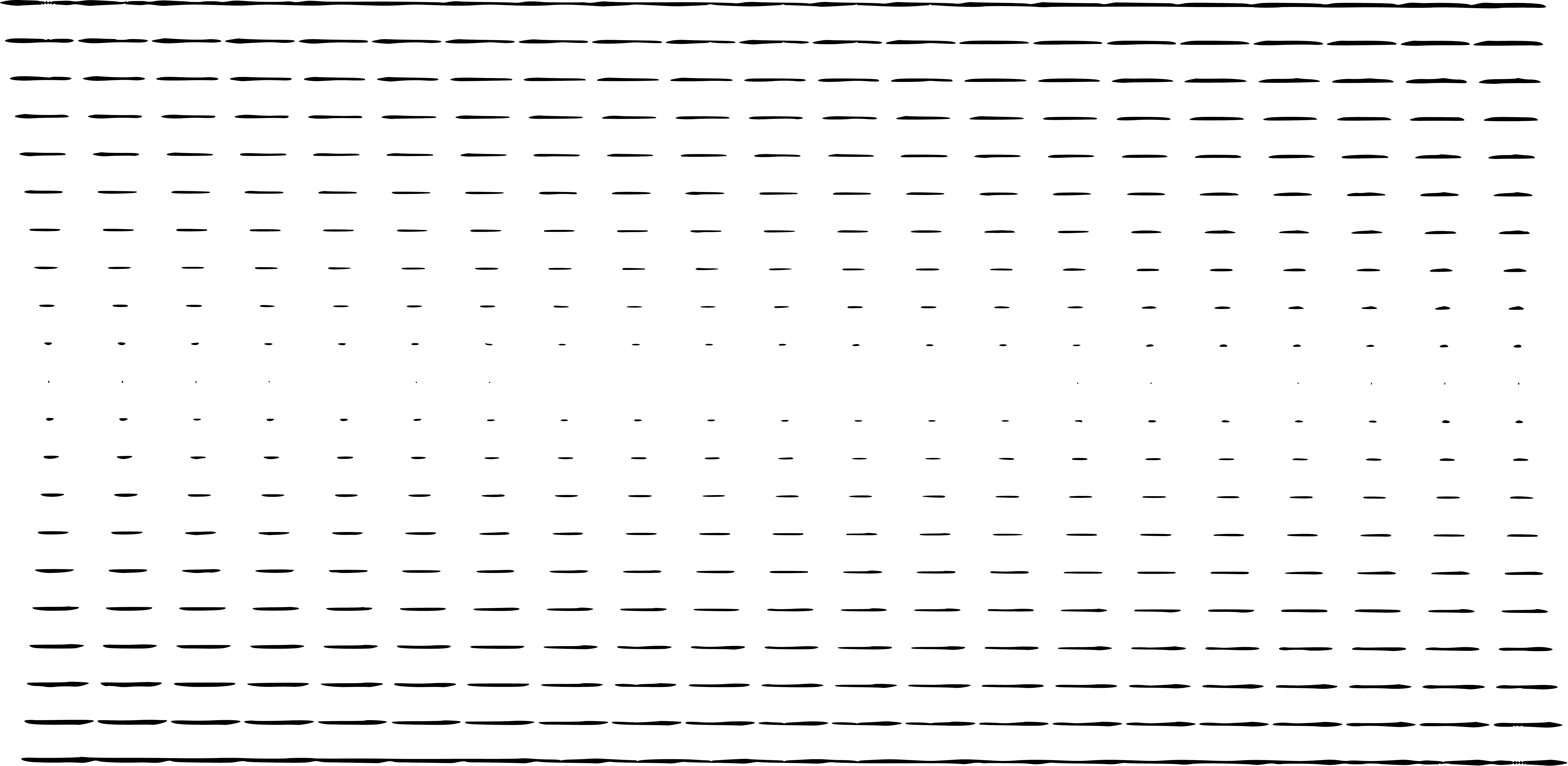}
\caption{\label{}Phase flow diagram on XOZ plane for Lorentz system}
\end{figure}

In addition, it can be seen that the streamline diagram of phase flow fits well with the traditional phase diagram of nonlinear system.
\begin{figure}[ht]
\includegraphics[scale=0.23]{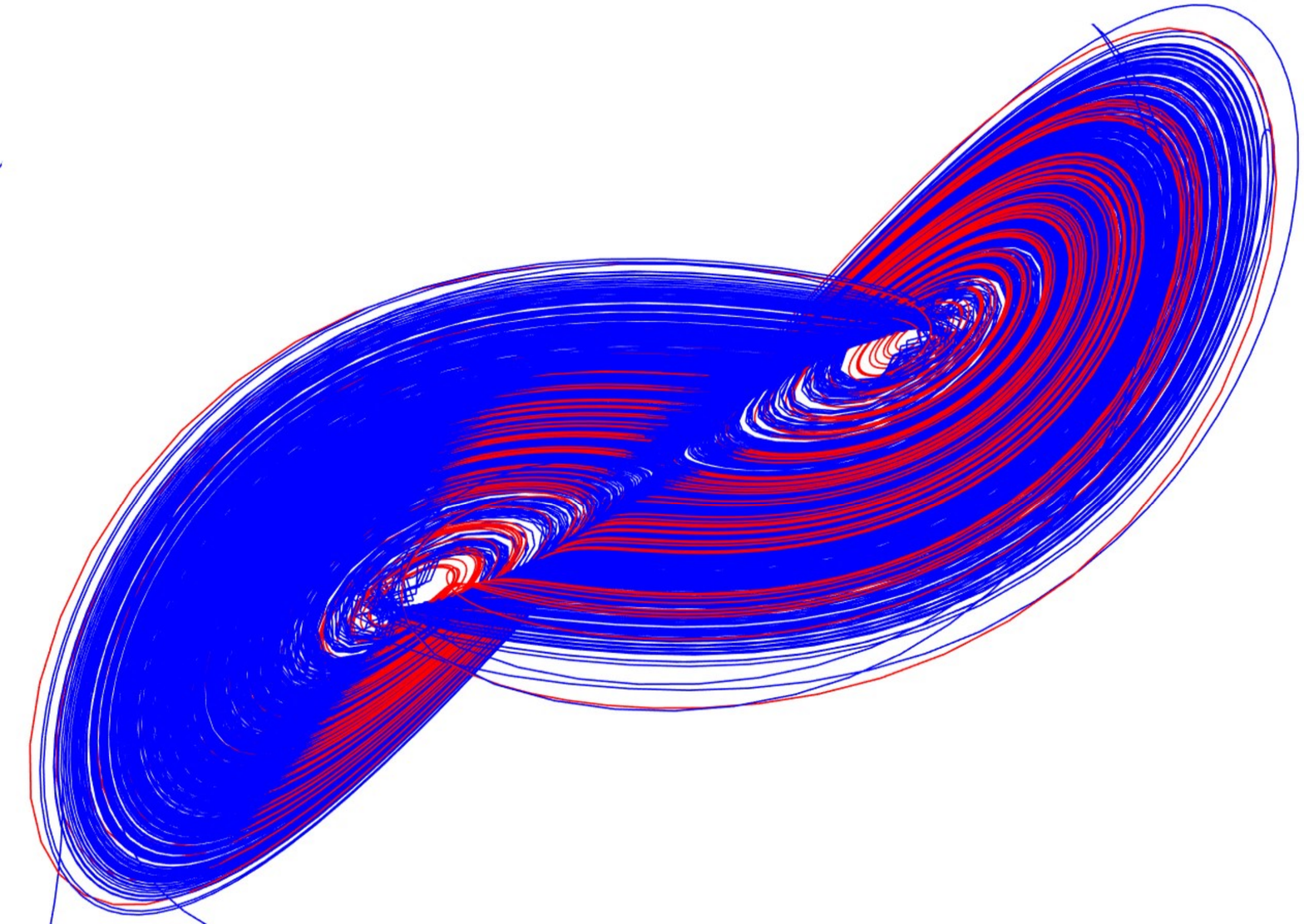}
\caption{\label{}Comparison between streamline diagram (blue) of phase flow and traditional phase diagram (red) for Lorentz system}
\end{figure}

Streamline, as one of the physical quantities of fluid mechanics, refers to the curve which is tangent to the velocity vector at every point in the flowfield. It is a curve composed of different fluid particles at the same time and reveals the spatial distribution law of velocity, which coincides with the relationship between momentum and position in phase diagram. It inspires a new method to solve PDEs and will be discussed in the future because it is not related to the content discussed in this paper.

From FIG. 1 to FIG. 3 provide a good foundation for the introduction of nonlinear analysis into field theory. In the history of conservative field research, firstly, Lagrange \cite{Fukushima2017} described the gravitational field with gravitational potential. The gravity at any point is equal to the negative gradient of the gravitational potential at that point. Then, Laplace \cite{Ebert2018} gives the gravitational potential equation in rectangular coordinates. Laplace \cite{Ebert2018} assumes that the Laplace equation holds when the attractive particle is in the body. Finally, Poisson \cite{Abadias2020} fixed the error. Poisson \cite{Abadias2020} pointed out that if the point is in the interior of the attractor, the Laplace equation should be changed to Poisson equation in practical physics research.
Therefore, we can understand that potential is the reason to promote the development of conservative field \cite{Greengard2018}. As we all know, nonlinear field includes conservative field and dissipative field. According to the expansion of the conservative field researches, we believe that the nonlinear field also has a “push potential" to promote its development. In order to study this “push potential”, here is the concept of generalized potential.
\begin{definition}
For $\Omega \subset \mathbb{R}^{n}$, $M=\left(\Omega,\mathscr{S}\right)$, existence tensor field $\sigma_{ij}\in \mathscr{T}_{M}$. if conjugate gradient of vector $\varphi_{i}\in \mathscr{T}_{M}$ satisfy condition $\sigma_{ij}=\varphi_{i}\partial_{j}$, the vector $\varphi_{i}$ is generalized potential of tensor field $\sigma_{ij}$.
\end{definition}

In order to explain why conjugate gradient is used in the definition, lemma 1 need to be studied.
\begin{lemma}
The curl of vector’s gradient is zero, and the curl of vector’s conjugate gradient is not zero.
\end{lemma}
\begin{proof}[Proof]
Due to arbitrary vector $\varphi_{i}\in \mathscr{T}_{M}$, $M=\left(\Omega,\mathscr{S}\right)$, the curl of vector s’ gradient is
$$\bm{e}^{k}\partial_{k}\times \bm{e}^{i}\bm{e}^{j}\partial_{i}\varphi_{j}=\left(\bm{e}^{k}\times \bm{e}^{i}\right)\bm{e}^{j}\varphi_{j,ki}=\epsilon^{ki}_{\bullet \bullet m}\bm{e}^{m}\bm{e}^{j}\varphi_{j,ki}$$

For partial derivatives, the order of derivation is transformed, and its value does not change.
$$\epsilon^{ki}_{\bullet \bullet m}\bm{e}^{m}\bm{e}^{j}\varphi_{j,ki}=\epsilon^{ki}_{\bullet \bullet m}\bm{e}^{m}\bm{e}^{j}\varphi_{j,ik}$$

By converting the index again, we can obtain
$$\epsilon^{ki}_{\bullet \bullet m}\bm{e}^{m}\bm{e}^{j}\varphi_{j,ik}=\epsilon^{ik}_{\bullet \bullet m}\bm{e}^{m}\bm{e}^{j}\varphi_{j,ik}=-\epsilon^{ki}_{\bullet \bullet m}\bm{e}^{m}\bm{e}^{j}\varphi_{j,ki}=0$$

It is proved that the curl of vector’s gradient is zero. Then, assuming that the curl of the conjugate gradient of the vector is zero and always holds
$$\bm{e}^{k}\partial_{k}\times \bm{e}^{i}\bm{e}^{j}\varphi_{i}\partial_{j}\equiv -\bm{e}^{k}\partial_{k}\times \bm{e}^{i}\bm{e}^{j}\varphi_{i}\partial_{j}$$

We can get the following formula
$$\bm{e}^{k}\partial_{k}\times \bm{e}^{i}\bm{e}^{j}\varphi_{i}\partial_{j}=\left(\bm{e}^{k}\times \bm{e}^{i}\right)\bm{e}^{j}\varphi_{i,kj}=\epsilon^{ki\bullet}_{\bullet \bullet m}\bm{e}^{m}\bm{e}^{j}\varphi_{i,kj}$$

According to the rule of index conversion, we can obtain
$$-\epsilon^{ki\bullet}_{\bullet \bullet m}\bm{e}^{m}\bm{e}^{j}\varphi_{i,kj}=\epsilon_{m \bullet \bullet}^{\bullet i k}\bm{e}^{m}\bm{e}^{j}\varphi_{i,kj}=\epsilon_{k \bullet \bullet}^{\bullet i m}\bm{e}^{k}\bm{e}^{j}\varphi_{i,mj}$$

Combine the above two formulas, we can get a funny conclusion
$$k\equiv m$$

According to the definition of tensor index, it is obvious that the curl of vector’s conjugate gradient is not zero.
\end{proof}

Thus it can be seen that the conjugate gradient is different from the gradient. The gradient of the vector has the expression problem of rotating field. However, the conjugate gradient represents the second-order tensor, which can be homeomorphic mapping with the original tensor. Therefore, any tensor can be expressed as conjugate gradient form of vector.

In order to establish the relationship between conjugate gradient and gradient, lemma 2 should be studied.
\begin{lemma}
The divergence of the conjugate gradient of an arbitrary vector is equal to the gradient of the divergence of this vector.
\end{lemma}
\begin{proof}[Proof]
Due to arbitrary vector $\varphi_{i}\in \mathscr{T}_{M}$, $M=\left(\Omega,\mathscr{S}\right)$, the divergence of the conjugate gradient of the vector is 
$$\bm{e}^{k}\partial_{k} \cdot \bm{e}^{i}\bm{e}^{j}\varphi_{i} \partial_{j}=\bm{e}^{j}\delta^{ki}\varphi_{i,jk}=\bm{e}^{j}\varphi_{i,ji}$$
and the gradient of the divergence of vector $\varphi_{i}$ is
$$\bm{e}^{j}\partial_{j}\left(\bm{e}^{k} \partial_{k} \cdot \bm{e}^{i}\varphi_{i}\right) =\bm{e}^{j}\delta^{ki}\varphi_{i,kj}=\bm{e}^{j}\varphi_{i,ij}$$
Because the order of derivation of partial derivatives can be interchanged arbitrarily, it can be obtained that
$$\bm{e}^{j}\varphi_{i,ji}=\bm{e}^{j}\varphi_{i,ij}$$

It can be concluded that the divergence of the conjugate gradient of an arbitrary vector is equal to the gradient of the divergence of this vector.
\end{proof}
In the field of classical mechanics, Newton's second law is usually considered when describing the relationship between the change of motion state and the force acting on it. So we can get the most important theorem of this paper.
\begin{theorem}
For stress tensor field $\sigma_{ij} \in \mathscr{T}_{M}$, $M=\left(\Omega,\mathscr{S}\right)$. The generalized potential of the stress field is $\varphi_i \in \mathscr{T}_{M}$, then the evolution law of the stress field satisfies
\begin{equation}
\varphi_{i,\bullet j}^{\bullet, i \bullet}=\rho \frac{du_j}{dt}
\end{equation}
where $\rho$ is the density of nonlinear system, $t$ is developing time of the nonlinear system and $u_j$ is the velocity vector of nonlinear system.
\end{theorem}

\begin{proof}[Proof]
Generally, an object is subjected to two kinds of force, surface force $F_j^{plane}$ and mass force $F_j^{mass}$. Hypothetically, if the force caused by mass is equivalently distributed to the surface, then it can be mathematically equivalent to sum surface force $F_j$.
$$F_j=F_j^{plane}+F_j^{mass}$$

According to Gauss theorem, above force can be described by stress state, $n^{i}$ is the surface normal vector.
$$F_{j}=\int_{\partial \tau}n^{i}\sigma_{ij} d\left(\partial \tau\right)=\int_{\tau}\sigma_{ij,\bullet}^{\bullet \bullet, i}d\tau$$
so the force per unit volume can be expressed as
$$f_{j}=\frac{dF_{j}}{d\tau}=\sigma_{ij,\bullet}^{\bullet \bullet, i}$$

Assuming that the generalized potential of the stress field is $\varphi_i \in \mathscr{T}_{M}$, according to lemmas 1 and 2, it is easy to obtain the following expression
$$\sigma_{ij,\bullet}^{\bullet \bullet, i}=\varphi_{i,j \bullet}^{\bullet,\bullet i}=\varphi_{i,\bullet j}^{\bullet, i \bullet}$$

Newton's second law with the force per unit volume in Mechanics can be expressed as:
$$f_{j}=\rho \frac{du_{j}}{dt}$$

Combining all of the above formulas, we can obtain the final result
$$\varphi_{i,\bullet j}^{\bullet, i \bullet}=\rho \frac{du_j}{dt}$$

In this way, the theory has been proved.
\end{proof}

This theorem reveals the spatiotemporal evolution law of nonlinear systems in classical mechanics. According to the famous formula in field theory. It can be simply understood that the sum of Laplacian operator and double curl operator is the operator on the left of formula (3) as shown in following
$$\varphi_{i,\bullet j}^{\bullet, i \bullet}=\varphi_{j,i\bullet}^{\bullet,\bullet i}+\epsilon_{j\bullet \bullet}^{\bullet n m}\epsilon_{m \bullet \bullet}^{\bullet k i}\varphi_{i,nk}$$

In the previous introduction, it has been explained that the physical field represented by the Laplacian operator is a conservative field. Therefore, theorem 2 can be obtained and the proof is omitted.
\begin{theorem}
In the framework of classical mechanics, nature exists a generalized potential $\varphi_i$ in the nonlinear system. The generalized potential is the fundamental reason that drives the development of nonlinear systems. The spatial evolution law of the generalized potential is
\begin{equation}
\varphi_{j,i\bullet}^{\bullet,\bullet i}+\epsilon_{j\bullet \bullet}^{\bullet n m}\epsilon_{m \bullet \bullet}^{\bullet k i}\varphi_{i,nk}=\rho \frac{du_j}{dt}
\end{equation}
where $\rho$ is the density of nonlinear system, $t$ is developing time of the nonlinear system and $u_j$ is the velocity vector of nonlinear system.
\end{theorem}

In addition, we can get two important corollaries, which reveal conditions for the occurrence of conservative system and pure dissipative system.
\begin{inference}
The conservative field satisfies the following relations, the Laplacian operator of the generalized potential is a constant value, and the acceleration drives the generalized potential to be spatially distributed in the form of double curl operator.
\begin{equation}
L_{j}=\varphi_{j,i\bullet}^{\bullet,\bullet i}, L_{j}\in \mathscr{T}_{M}
\end{equation}

At this time, the space-time evolution law can be transformed into
\begin{equation}
\epsilon_{j\bullet \bullet}^{\bullet n m}\epsilon_{m \bullet \bullet}^{\bullet k i}\varphi_{i,nk}=\rho \frac{du_j}{dt}-L_j
\end{equation}
where $\rho$ is the density of nonlinear system, $t$ is developing time of the nonlinear system and $u_j$ is the velocity vector of nonlinear system.
\end{inference}
\begin{inference}
The pure-dissipation field satisfies the following relations, the double curl operator of the generalized potential is a constant value, and the acceleration drives the generalized potential to be spatially distributed in the form of Laplacian.
\begin{equation}
D_{j}=\epsilon_{j\bullet \bullet}^{\bullet n m}\epsilon_{m \bullet \bullet}^{\bullet k i}\varphi_{i,nk}, D_{j}\in \mathscr{T}_{M}
\end{equation}

At this time, the space-time evolution law can be transformed into
\begin{equation}
\varphi_{j,i\bullet}^{\bullet,\bullet i}=\rho \frac{du_j}{dt}-D_j
\end{equation}
where $\rho$ is the density of nonlinear system, $t$ is developing time of the nonlinear system and $u_j$ is the velocity vector of nonlinear system.
\end{inference}

Formula (5) reveals the basic evolution law of general chaotic system, and formula (6) and formula (7) are two limits of chaotic motion. Therefore, the following physical quantities can be defined to evaluate the degree of conservation and dissipation in any chaotic system.
\begin{definition}
For $\Omega \subset \mathbb{R}^{n}$, $M=\left(\Omega,\mathscr{S}\right)$, existence tensor field $\sigma_{ij}\in \mathscr{T}_{M}$. $\mathcal{L}$ is the Laplace operator and $\mathcal{D}$ is the double curl operator. $L_j \in \mathscr{T}_{M}$ is the conservative vector field and $D_j \in \mathscr{T}_{M}$ is the dissipation vector field for the original tensor field $\sigma_{ij}$. The development speed of complex system is $u_j \in \mathscr{T}_{M}$ and the generalized potential is $\varphi_j \in \mathscr{T}_{M} $, then these nonlinear development law of complex system is as follows:
\begin{equation}
\rho\frac{du_j}{dt}=L_j+D_j, L_j=\mathcal{L}\left(\varphi_j\right), D_j=\mathcal{D}\left(\varphi_j\right)
\end{equation}
where $\rho$ is the density of complex system, $t$ is developing time of the complex system
\end{definition}

Finally, an example is given to illustrate the influence of generalized potential on nonlinear system. If the distribution of generalized potential in three-dimensional space is in the following form
\begin{equation}
\begin{array}{c}
\varphi_{1}=\sigma\left(x_{2}-x_{1}\right)\\
\varphi_{2}=rx_{1}-x_{2}-x_{1}x_{3}\\
\varphi_{3}=x_{1}x_{2}-bx_{3}
\end{array}
\end{equation}

It is noticed that above equations look like Lorentz system, but they are the spatial distribution of the generalized potential, not the spatial distribution of the original field, so it is not Lorentz system. Let's name it the new-1 nonlinear system. Its generalized potential space distribution is shown in FIG. 1 to FIG. 3. The distribution of dissipation vector is shown in FIG. 4 to FIG. 7.
\begin{figure}[ht]
\includegraphics[scale=0.21]{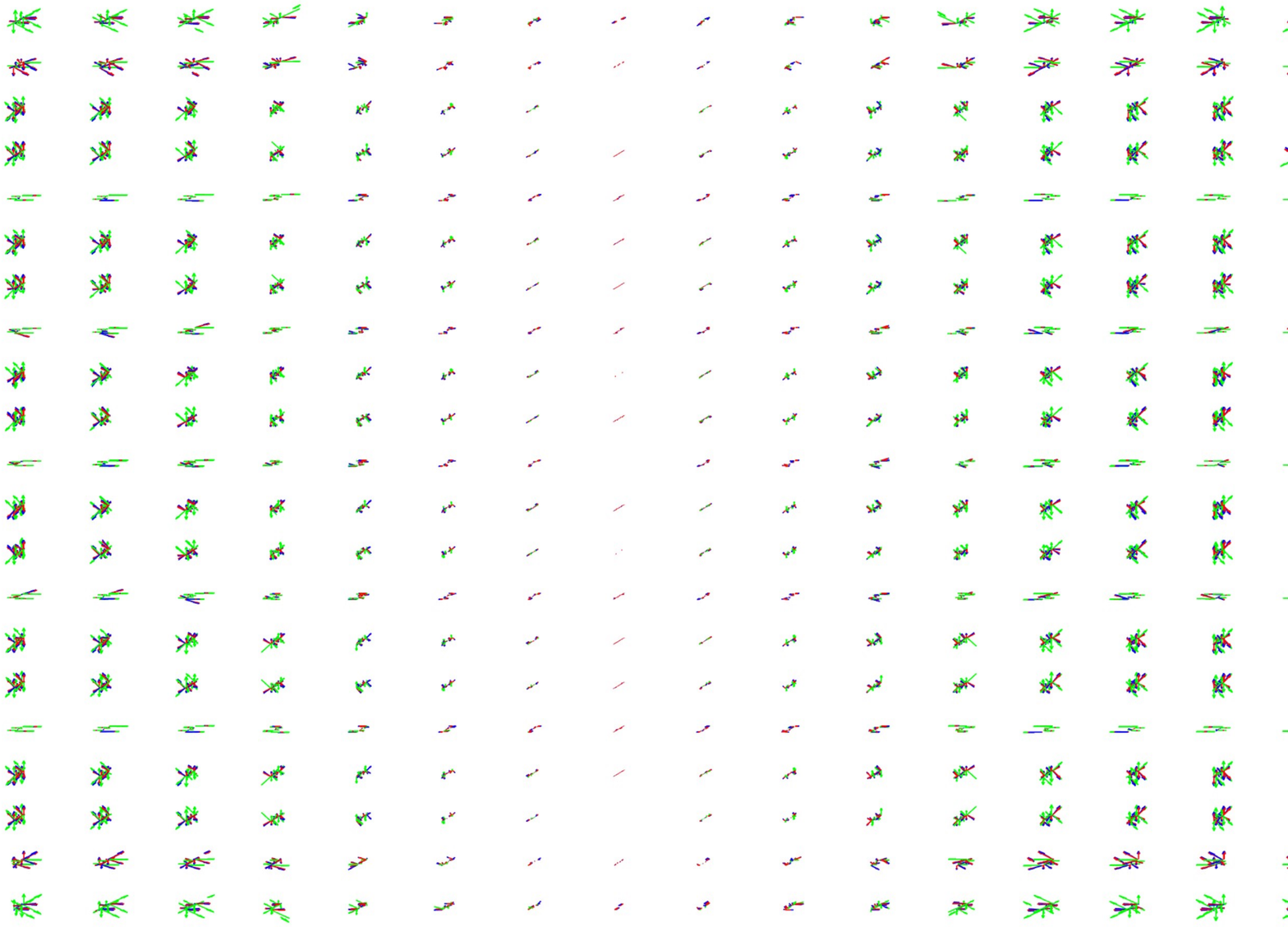}
\caption{\label{}Dissipation (green) and conservative (red) vector field distribution and acceleration field (blue) on XOY plane for nonlinear system New-1}
\end{figure}
\begin{figure}[ht]
\includegraphics[scale=0.21]{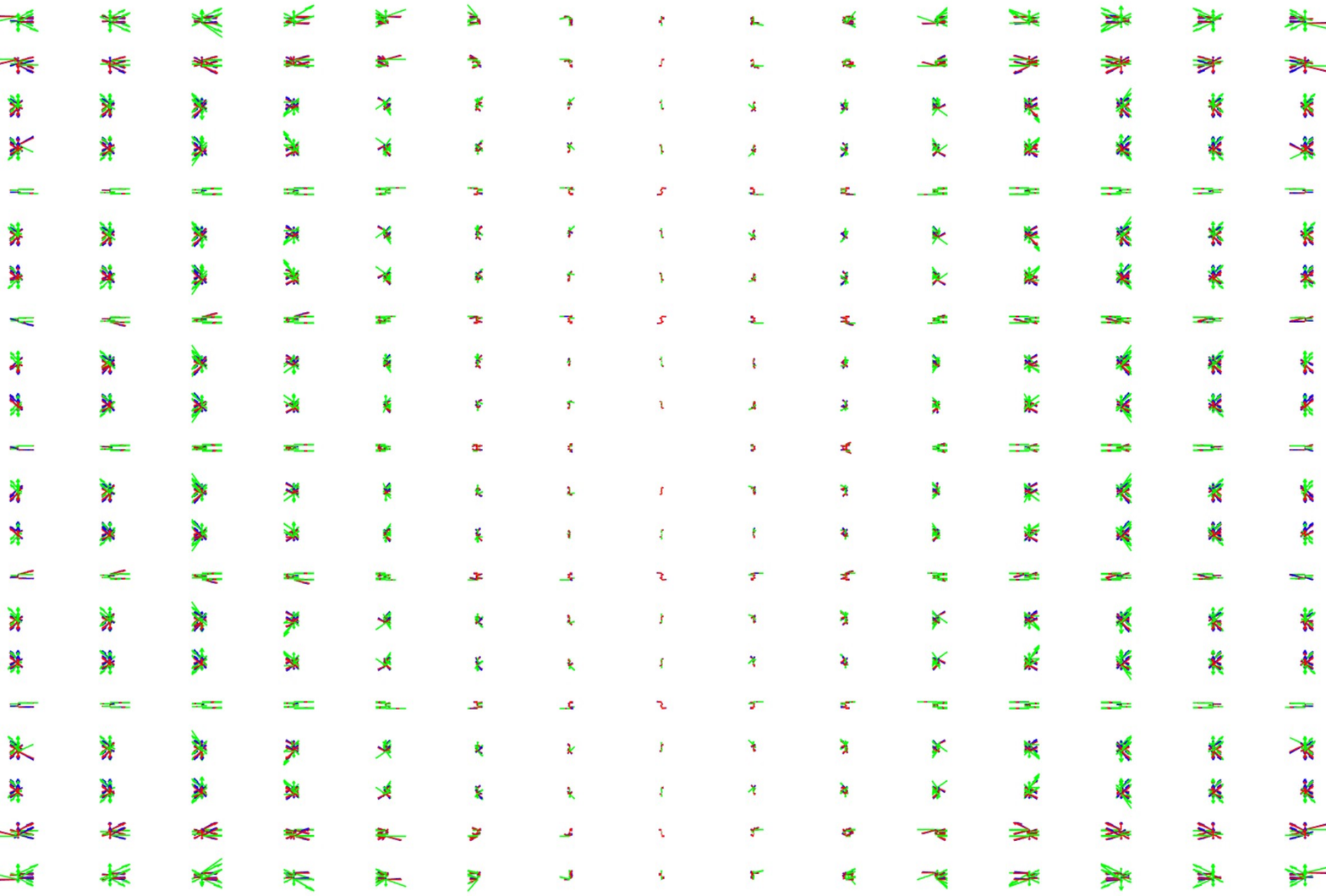}
\caption{\label{}Dissipation (green) and conservative (red) vector field distribution and acceleration field (blue) on XOZ plane for nonlinear system New-1}
\end{figure}
\begin{figure}[ht]
\includegraphics[scale=0.21]{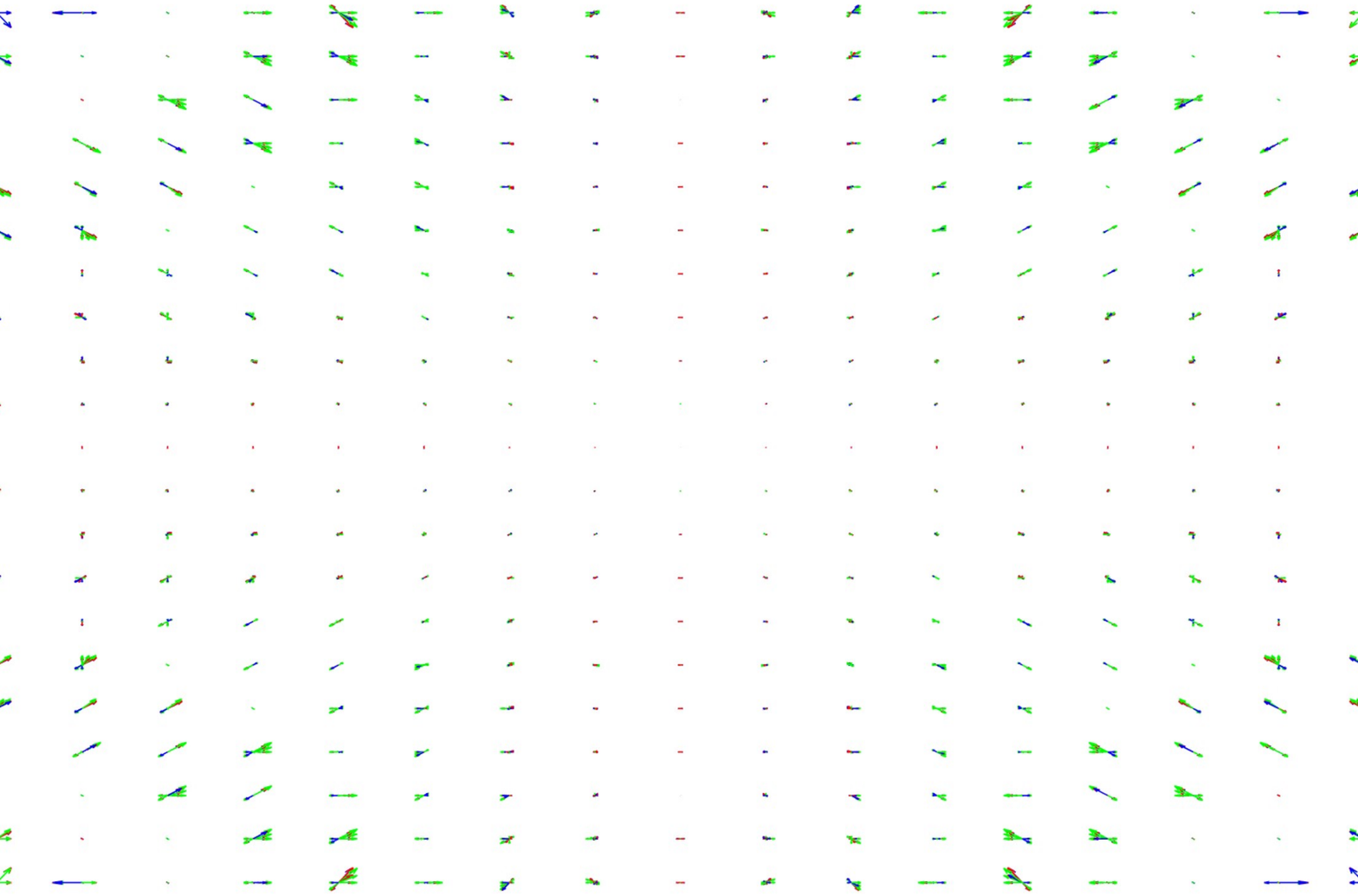}
\caption{\label{}Dissipation (green) and conservative (red) vector field distribution and acceleration field (blue) on YOZ  plane for nonlinear system New-1}
\end{figure}

It can be seen from FIG. 4, FIG. 5 and FIG. 6 that the small disturbance at the center point will bring vibration of the distant space acceleration, which reveals the famous "Butterfly Effect" \cite{Cai2020}. For the new-1 nonlinear system, the maximal amplifying coefficient is between 115 and 155. In fact, due to the superposition of vectors, the small disturbances in different places will be integrated to balance. Therefore, if we want to calculate the disturbance amplifying coefficient in space, we need to consider all the small disturbance values and stack them. The cancellation of the nonlinear interaction can be demonstrated in the experiment of turbulent interaction changing into laminar flow published in 2019 \cite{Kuchnen2019}. In addition, it can be seen from the streamline diagram FIG. 7 that the dissipative structure of nonlinear disturbance is reciprocating, and it can be seen that system will develop from one equilibrium state to another. Here, the dissipative structure theory \cite{Prigogine1965} proposed by Ilya Prigogine is verified mathematically and physically. 
So we can get the following general law in physics. 
\begin{theorem}
All nonlinear processes spontaneously undergo a conservative and dissipative reciprocating process. The generalized potential is distributed in the space in the form of divergence gradient operator, where the spatial distribution of the generalized potential on Laplace operator represents the conservative behavior of the nonlinear system, and the spatial distribution on the double curl operator of the generalized potential represents the dissipative behavior of the nonlinear system, until the kinetic energy of the nonlinear system become zero and chaos stops. 
\end{theorem}

\begin{figure}[ht]
\includegraphics[scale=0.15,angle=90]{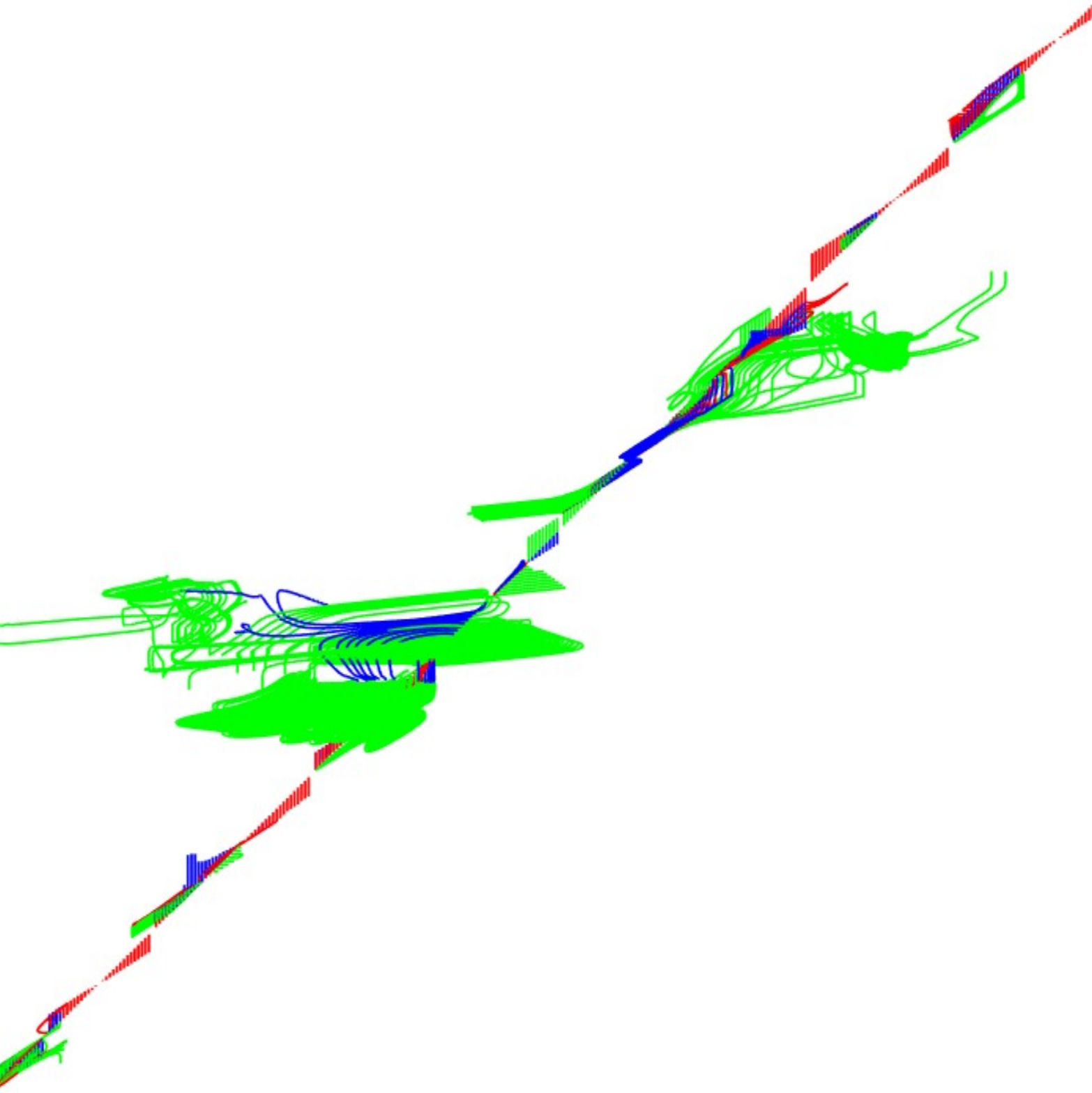}
\caption{\label{}Streamlines of dissipation (green) and conservative (red) vector field distribution and acceleration field (blue) for nonlinear system New-1}
\end{figure}

Nonlinearity is the essential law of nature, so it is not only a physical law, but also a new natural philosophcal methodology. Due to the regular generalized potential distribution (Fig.1 and Fig.2) will lead to random and irregular distribution of kinematic physical quantities in space (Fig.4, Fig.5 and Fig.6) , the double curl operator of complex system will lead to the failure of Laplacian determinism, Therefore, generalized potential and its principle can not only be applied to fluid mechanics, solid mechanics, electrodynamics, even economics, political science, history and other disciplines, but also contribute to development of natural philosophy principle.

In summary, mathematically speaking, the fundamental cause of nonlinear phenomena is the double curl operator of the generalized potential. Due to the influence of different initial value and different boundary conditions, the spatial distribution of generalized potential is different, and the development law of nonlinear system has also undergone great changes. Its application scenarios could be different physical scenarios of analytical mechanics, fluid mechanics, solid mechanics and electrodynamics research. Therefore, the generalized potential can effectively unify classical mechanics. In addition, generalized potential can be used as an effective method to solve nonlinear problems in engineering. This is an open and challenging new field. Determining generalized potentials with different boundary conditions and different initial value will be the focus of future research in this field.

\bibliographystyle{unsrt}
\bibliography{main}

\begin{thebibliography}{10}

\bibitem{Kuchnen2019}
J.~Kuchnen, B.~Song, D.~Scarselli, and et.al.
\newblock Destabilizing turbulence in pipe flow.
\newblock {\em Nature physics}, 14:386--390, 2019.

\bibitem{Shen2020}
J.~Shen.
\newblock A recursive bifurcation model for early forecasting of covid-19 virus
  spread in south korea and germany.
\newblock {\em Science Report}, 10(20776), 2020.
\newblock {DOI:10.1038/s41598-020-77457-5}.

\bibitem{Manda2020}
J.~L. Manda, M.~Dino, K.~Joseph, and et.al.
\newblock Urbanization and market integration have strong, nonlinear effects on
  cardiometabolic health in the turkana.
\newblock {\em Science Advance}, 6(43), 2020.
\newblock {DOI:10.1126/sciadv.abb1430}.

\bibitem{Barnett2020}
L.~R. Barnett, J.~D. Charman, and et. al.
\newblock Nonlinear landscape and cultural response to sea-level rise.
\newblock {\em Science Advance}, 6(43), 2020.
\newblock {DOI:10.1126/sciadv.abb6376}.

\bibitem{Buckmaster2019}
T.~Buckmaster and V.~Vicol.
\newblock Non-uniqueness of weak solutions to the navier-stokes equation.
\newblock {\em Annals of mathematics}, 189(1), 2019.
\newblock {DOI:10.4007/ annals. 2019. 189.1.3}.

\bibitem{Hamzi2019}
B.~Hamzi and E.~H. Abed.
\newblock Local modal participation analysis of nonlinear systems using
  poincaré linearization.
\newblock {\em Nonlinear Dynamics}, 99:803–811, 2019.
\newblock {DOI:10.1007/s11071-019-05363-1}.

\bibitem{Kuznetsov2020}
N.~V. Kuznetsov, T.~N. Mokaev, O.~A. Kuznetsova, and E.~V. Kudryashova.
\newblock The lorenz system: hidden boundary of practical stability and the
  lyapunov dimension.
\newblock {\em Nonlinear Dynamics}, 2020.
\newblock {DOI:10.1007/s11071-020-05856-4}.

\bibitem{Fukushima2017}
T.~Fukushima.
\newblock Lagrange's analytical mechanics (lagrange 1788) in original french
  text.
\newblock 05 2017.
\newblock {DOI:10.13140/RG.2.2.19275.59684}.

\bibitem{Ebert2018}
M.~R. Ebert and M.~Reissig.
\newblock Laplace equation - properties of solutions - starting point of
  elliptic theory.
\newblock {\em Springer}, 2018.

\bibitem{Abadias2020}
L.~Abadias, E.~G. Joséa, and C.~Lizama.
\newblock Poisson equation and discrete one-sided hilbert transform for
  $\left(c,\alpha\right)$ bounded operators.
\newblock {\em arXiv}, 2020.
\newblock {arXiv:2002.10122}.

\bibitem{Greengard2018}
L.~Greengard and S.~Jiang.
\newblock A new mixed potential representation for the equations of unsteady,
  incompressible flow.
\newblock {\em arXiv}, 2018.
\newblock {arxiv:1809.08442}.

\bibitem{Cai2020}
W.~Cai, B.~Ng, T.~Geng, L.~Wu, and A.~Santoso.
\newblock Publisher correction: Butterfly effect and a self-modulating el niño
  response to global warming.
\newblock {\em Nature}, 2020.
\newblock {DOI:10.1038/s41586-020-2958-5}.

\bibitem{Prigogine1965}
I.~Prigogine.
\newblock Herausgeber: Advances in chemical physics.
\newblock {\em Interscience Publishers, John Wiley and Sons, London‐New
  York}, (416), 1965.

\end{thebibliography}

\end{document}